\documentclass[aps,pra,showpacs,superscriptaddress,amssymb,twocolumn]{revtex4-1}
\usepackage{graphicx}
\usepackage{appendix} 
\usepackage{hyperref}
\usepackage{color}\hypersetup{colorlinks=true,citecolor=blue,linkcolor=black}
\usepackage{bm,amsmath}
\usepackage{amsthm}
\usepackage{dcolumn}
\newtheorem{thm}{Theorem}

\newtheorem{definition}{Definition}

\newcommand{\be}{\begin{equation}}
\newcommand{\ee}{\end{equation}}

\begin{document}

\rightline{\tt Quantum Science and Technology {\bf 5}, 03LT01 (2020)}

\vspace{0.2in}

\title{Rigorous measurement error correction}

\author{Michael R. Geller}
\affiliation{Center for Simulational Physics, University of Georgia, Athens, Georgia 30602, USA}
\date{\today}

\begin{abstract}
We review an experimental technique used to correct state preparation and measurement errors on gate-based quantum computers, and discuss its rigorous justification. Within a specific {\it biased quantum measurement} model, we prove that nonideal measurement of an arbitrary $n$-qubit state is equivalent to ideal projective measurement followed by a classical Markov process $\Gamma$ acting on the output probability distribution. Measurement errors can be removed, with rigorous justification, if $\Gamma$ can be learned and inverted.  We show how to obtain $\Gamma$ from gate set tomography (R. Blume-Kohout {\it et al.}, arXiv:1310.4492) and apply the rigorous correction technique to single IBM Q superconducting qubits.
\end{abstract}

\maketitle

\section{Introduction}

A well-known technique for mitigating state preparation and measurement (SPAM) errors is to measure the transition (overlap-squared) matrix $T$ between all initial and final classical states, and minimize $\| T \, p_{\rm corr} - p_{\rm noisy} \|_2^2$ subject to constraints $0 \le p_{\rm corr}(x) \le 1$ and $\| p_{\rm corr} \|_1 =1$ to correct subsequently measured probability distributions \cite{BialczakNatPhys10,NeeleyNat10,DewesPRL12,14114994,160304512,SongPRL17,181102292,190505720,180411326,TannuIEEE19,190411935,190503150,190708518,191000129,191001969,191113289,200104449,200109980}. Although there are practical limitations of the method that have prevented application to large registers, in this work we address a different issue: its rigorous justification. Suppose that, after a quantum computation, the probability distribution on one or more qubits has to be accurately estimated. For example, a recent experiment by Gong {\it et al.} \cite{181102292} applied the $T$ matrix method to mitigate SPAM errors for the estimation of the fidelity of a 12-qubit cluster state, which must exceed $\frac{1}{2}$ to certify genuine multipartite entanglement. Is correcting SPAM this way is reliable? How can we be sure that we are not corrupting a measured probability distribution?

We answer this question in two parts: First we define a specific nonideal quantum measurement model for which rigorous error correction is possible, at least in principle. Within this model, the error correction protocol is identical to that of the $T$ matrix method, except with a doubly stochastic matrix $\Gamma \neq T.$ Second, we discuss how to obtain $\Gamma$. In the unlikely context that classical states $\lbrace |x\rangle \rbrace_{x \in  \lbrace 0,1 \rbrace^n}$  can be prepared with negligible error, $\Gamma$ can be directly measured experimentally (in this case $\Gamma \! = \! T$). Usually, however, $\Gamma$ needs to be obtained by some other means, such tomography. Here we show how to obtain $\Gamma$ from gate set tomography \cite{BlumeKohout13104492,BlumeKohoutNat17} and apply the rigorous correction technique to single IBM Q superconducting qubits. The results of $\Gamma$ matrix measurement correction are also compared with that of the $T$ matrix, highlighting some deficiencies of the latter.

\section{Biased measurement model}

In this work we describe the nonideal measurement of a register of $n$ qubits by a set of $2^n$ positive operator-valued measure (POVM) elements 
\begin{equation}
\lbrace E_x \rbrace_{ x \in \lbrace 0,1 \rbrace^n}
\end{equation}
satisfying
\begin{equation}
\sum_{x \in \lbrace 0,1 \rbrace^n} E_x = I ,
\label{defPOVMnormalization}
\end{equation} 
and such that, in the ideal limit, 
\begin{equation}
E_x \rightarrow |x\rangle \langle x |
= |x_1\rangle \langle x_1 | \otimes |x_2\rangle \langle x_2 | \otimes \cdots \otimes |x_n\rangle \langle x_n |.
\end{equation} 
More precisely, we assume that each noisy measurement operator is closest in Frobenius norm to a unique classical-state projector $|x\rangle \langle x |$, and we label the measurement operator by that projector. Therefore, $E_x$ is the noisy version of the projector $|x\rangle \langle x |$.

\begin{definition}
Let $x \in \lbrace 0,1 \rbrace^n$ be a classical state and $\lbrace E_x \rbrace_x$ be a $2^n$-outcome POVM acting on $n$ qubits. We call $\lbrace E_x \rbrace_x$ a biased measurement model if every element $E_x $ is diagonal in the classical basis.
\end{definition}

The biased measurement model is rich enough to include the multiqubit crosstalk errors studied in \cite{BialczakNatPhys10,NeeleyNat10,200109980}. However it excludes measurement-frame orientation errors, such as a $z$ basis measurement with unintended tilt. The biased measurement model is expected to be a good approximation for dispersively measured superconducting qubits, where the dominant readout error comes from energy relaxation ($T_1$ decay) during qubit readout. Tomography data presented in Table \ref{gstSummaryTable} below is consistent with this expectation.

\begin{thm}
[Biased measurement]
Let $\lbrace E_x \rbrace_x$ be a biased measurement model for a register of $n$ qubits. Then nonideal quantum measurement (according to this model) of an arbitrary $n$-qubit state is equivalent to ideal projective measurement followed by a classical Markov process $\Gamma$ acting on the output probability distribution. 
The equivalence is illustrated in Fig.~\ref{CircuitIdentity}.
\label{BiasedMeasurementTheorem}
\end{thm}

\begin{proof}
Let $\rho$ be the state of the $n$-qubit register at the end of a circuit, just before measurement. The experimentally observed probability for observing classical state $x$ is
\begin{eqnarray}
p_{\rm noisy}(x) &=& {\rm Tr} \, (E_x \rho) \nonumber \\
&=& \sum_{x' \in \lbrace 0,1 \rbrace^n} \langle x' | E_x \rho | x' \rangle \nonumber \\
&=& \sum_{x' \in \lbrace 0,1 \rbrace^n} \Gamma(x|x') \,p_{\rm ideal}(x'),
\end{eqnarray}
where $p_{\rm ideal}(x) = \langle x | \rho | x\rangle$ is the uncorrupted probability distribution and 
\begin{equation}
\Gamma(x|x') := {\rm Tr}(E_x \, |x'\rangle \langle x'|)
\label{GammaDefinition}
\end{equation}
are the matrix elements of the Markov process. Regarding $p_{\rm noisy}(x)$ and $p_{\rm ideal}(x)$ as a vectors, the measurement error therefore acts via multiplication by the $2^n \! \times 2^n$ matrix $\Gamma$: 
\begin{equation}
p_{\rm ideal} \mapsto
p_{\rm noisy} = \Gamma \, p_{\rm ideal},
\end{equation}
as stated. $\Box$
\end{proof}

\begin{figure}
\includegraphics[width=7.0cm]{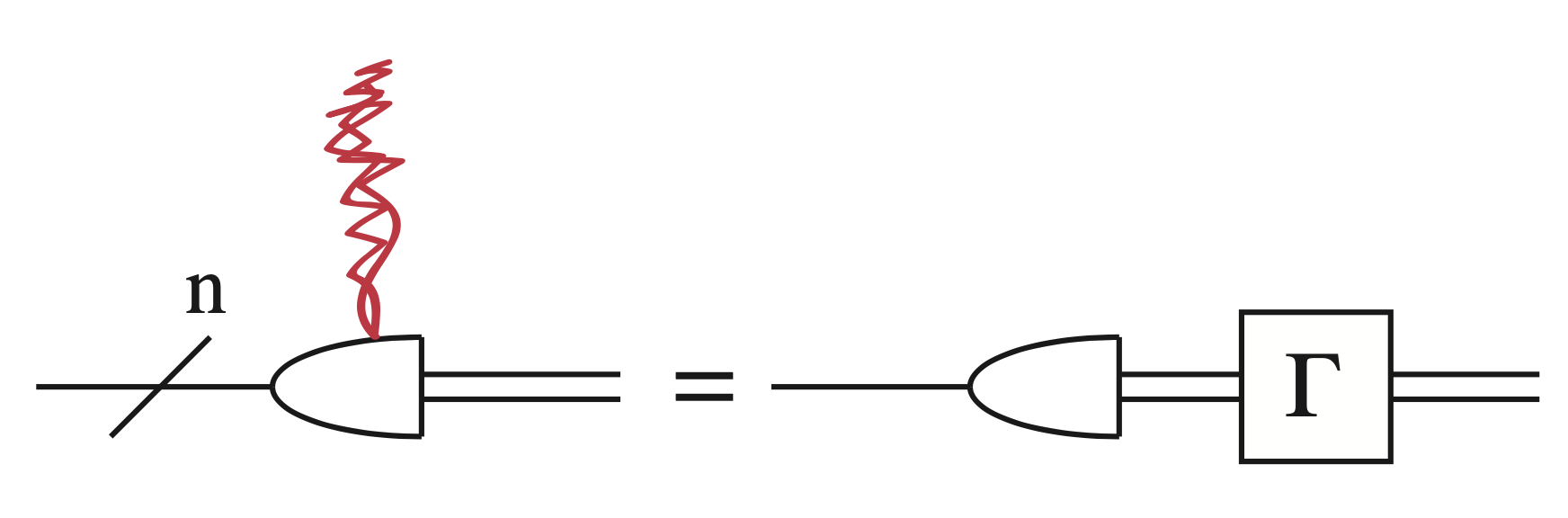} 
\caption{Circuit identity illustrating the biased measurement theorem. The measurement operation on the left side of the identity is noisy. The measurement on the right side is an ideal projective measurement. $\Gamma$ is an $2^n \! \times \! 2^n$ stochastic matrix acting on the output probability distribution.}
\label{CircuitIdentity}
\end{figure} 

\section{$\Gamma$ versus $T$}

The $\Gamma$ matrix used for rigorous measurement correction is defined in (\ref{GammaDefinition}).  The transition matrix can be written similarly as
\begin{equation}
T(x|x') = {\rm Tr}(E_x \, \rho_{x'}),
\label{defT}
\end{equation}
where $\rho_{x'}$ is the noisy implementation of the classical state $|x'\rangle \langle x'|$. Recall that the $T$ matrix is measured on an empty (identity) circuit. In the absence of any SPAM errors, both $T$ and $\Gamma$ are equal to the $2^n \! \times \! 2^n$ identity matrix $I$:
\begin{equation}
 T(x|x') = \Gamma(x|x') = \delta_{xx'}. \ \ \ ({\rm ideal})
\end{equation}
Comparing (\ref{GammaDefinition}) and  (\ref{defT}), we see that $\Gamma$ would equal $T$ if the classical states are prepared perfectly (and the biased measurement model is valid). 

The $T$ matrix  is {\it left stochastic}, consisting of columns of nonnegative real numbers summing to one,
\begin{equation}
\sum_x T(x|x') = {\rm Tr} \big( \sum_x E_x \, \rho_{x'} \big) = 1.
\label{defStochastic}
\end{equation}
This follows from (\ref{defPOVMnormalization}). The $\Gamma$ matrix is both left and right stochastic, also having rows summing to one,
\begin{equation}
\sum_x \Gamma(x|x') = \sum_{x'} \Gamma(x|x') = 1.
\label{defDoubleStochastic}
\end{equation}
However the matrix inverses $T^{-1}$ and $\Gamma^{-1}$ (if they exist) are not necessarily stochastic.

\section{Measurement error correction}
\label{measurement error correction}

Within the biased measurement model, Theorem \ref{BiasedMeasurementTheorem} shows that if $\Gamma$ can be measured and inverted, the measurement errors can be rigorously corrected by applying $\Gamma^{-1}$ to $p_{\rm noisy}(x)$. However the existence of $\Gamma^{-1}$ is not guaranteed by the biased measurement model. We treat the cases of invertible and singular $\Gamma$ separately. 

\subsection{Invertible $\Gamma$}

If $\Gamma$ is nonsingular the corrected probability distribution is 
\begin{equation}
p_{\rm corr} = \Gamma^{-1} \, p_{\rm noisy}.
\label{defGammaErrorCorrection}
\end{equation}
If the actual measurement operators are diagonal and $\Gamma$ is known precisely, then $p_{\rm corr}$ is the probability distribution that would be observed with ideal projective measurement. This means that  $p_{\rm corr}$ is automatically a proper probability distribution satisfying 
$0 \le p_{\rm corr}(x) \le 1$ and $\| p_{\rm corr} \|_1 =1$.

A weakness of this approach is that sampling errors in $p_{\rm noisy}$ can lead to an unphysical $p_{\rm corr}$, in particular, to negative values. However $\| p_{\rm corr} \|_1 =1$ is guaranteed for any normalized $p_{\rm noisy}$ because $ \Gamma$ is a stochastic matrix, which conserves the 1-norm when acting on vector with nonnegative entries (such as a probability distibution). If negative values are obtained, an alternative to (\ref{defGammaErrorCorrection}) must  be applied. The most direct is to calculate
\begin{equation}
p_{\rm corr} = {\rm argmin}_p \, \| \Gamma \, p - p_{\rm noisy} \|_2^2
\label{defLeastSquares}
\end{equation}
subject to the physicality constraints $0 \le p(x) \le 1$ and $\| p \|_1 =1$. This alternative is optimal in the sense that any protocol returning a physical $p_{\rm corr}$ must be equivalent to (\ref{defLeastSquares}) or must not fit the data as well (in 2-norm).

\subsection{Singular $\Gamma$}

In this case rigorous measurement correction is not possible. However the optimal alternative (\ref{defLeastSquares}) can still be applied in this case. While not exactly reversing the effect of the measurement error, the protocol (\ref{defLeastSquares}) is still likely to be a useful heuristic error mitigation technique.

\section{Measuring $\Gamma$}

Although the $T$ matrix can be directly measured experimentally, the $\Gamma$ matrix cannot (unless state preparation errors are completely negligible).  Here we discuss and apply an approach for estimating $\Gamma$ with gate set tomography \cite{BlumeKohout13104492,BlumeKohoutNat17} using our BQP data acquisition software combined with pyGSTi \cite{pygsti}. Gate set tomography works by measuring a large batch of circuits, designed to amplify errors by repeating carefully chosen germ sequences, and then performing a maximum likelihood fit of the resulting data to a set of stationary models for the prepared state, gates, and measurement operators. We measured 589 distinct circuits up to length 16 (18 including fiducials), and used the robust CPTP estimates to mitigate non-Markovian effects such as drift and leakage. The gate set included imperfect $\pi/2$ rotations about $x$ and $y$. To suppress statistical errors all circuits were measured with 32k measurement samples (each distinct circuit was measured four times with 8k samples). The $\Gamma$ matrix was then evaluated directly from (\ref{GammaDefinition}). Our results for the IBM Q chip ibmq\_essex are summarized in Table \ref{gstSummaryTable}. In all cases the estimated  $E_0$ and $E_1$ are very nearly diagonal, validating the biased measurement model.

\begin{widetext}

\begin{table}[htb]
\centering
\caption{Reconstructed $|0\rangle$ state and two-outcome POVM estimated by gate set tomography for the IBM Q device ibmq\_essex \cite{dataTakingNote}. $T$ is the concurrently measured transition matrix (\ref{defT}). $\Gamma$ is calculated from (\ref{GammaDefinition}).}
\begin{tabular}{|c|c|c|c|c|}
\hline
ibmq\_essex qubit& $\rho_{|0\rangle}$ & $E_0$ & $T$ &  $\Gamma$ \\
\hline 
$Q_0$   
& $\begin{pmatrix} 0.9792 & -0.0390 \\ -0.0390 & 0.0208 \\ \end{pmatrix}$   
& $\begin{pmatrix} 1.0000 & -0.0020 \\ -0.0020 & 0.0462 \\ \end{pmatrix}$ 
& $\begin{pmatrix} 0.9798 & 0.0606 \\ 0.0202 & 0.9394 \\ \end{pmatrix}$ 
& $\begin{pmatrix} 1.0000 & 0.0462 \\ 0.0000 & 0.9538 \\ \end{pmatrix}$  
\\ 
\hline 
$Q_1$ 
& $ \begin{pmatrix} 0.9703 & -0.0518 \\ -0.0518 & 0.0297 \\ \end{pmatrix} $   
& $\begin{pmatrix} 1.0000 & 0.0028 \\ 0.0028 & 0.0718 \\ \end{pmatrix} $ 
& $ \begin{pmatrix} 0.9793 & 0.0692 \\ 0.0207 & 0.9308 \\ \end{pmatrix}$ 
& $ \begin{pmatrix} 1.0000 & 0.0718 \\ 0.0000 & 0.9282 \\ \end{pmatrix} $ 
\\
\hline
$Q_2$ 
& $ \begin{pmatrix} 0.9928 & -0.0037 \\ -0.0037 & 0.0072 \\ \end{pmatrix} $   
& $ \begin{pmatrix} 1.0000 & -0.0005 \\ -0.0005 & 0.0650 \\ \end{pmatrix} $ 
& $ \begin{pmatrix} 0.9928 & 0.0801 \\ 0.0072 & 0.9199 \\ \end{pmatrix} $ 
& $ \begin{pmatrix}1.0000 & 0.0650 \\ 0.0000 & 0.9350 \\ \end{pmatrix} $ 
\\
\hline
$Q_3$ 
& $ \begin{pmatrix} 0.9839 & -0.0548 \\ -0.0548 & 0.0161 \\ \end{pmatrix} $   
& $ \begin{pmatrix} 1.0000 & -0.0004 \\ -0.0004 & 0.0413 \\ \end{pmatrix} $ 
& $ \begin{pmatrix} 0.9837 & 0.0597 \\ 0.0163 & 0.9403 \\ \end{pmatrix} $ 
& $ \begin{pmatrix} 1.0000 & 0.0413 \\ 0.0000 & 0.9587 \\ \end{pmatrix} $ 
\\
\hline
$Q_4$ 
& $ \begin{pmatrix} 0.8590 & -0.0056 \\ -0.0056 & 0.1410 \\ \end{pmatrix} $   
& $ \begin{pmatrix} 0.9539 & -0.0022 \\ -0.0022 & 0.0680 \\ \end{pmatrix} $ 
& $ \begin{pmatrix} 0.8508 & 0.1599 \\ 0.1492 & 0.8401 \\ \end{pmatrix}$ 
& $ \begin{pmatrix} 0.9539 & 0.0680 \\ 0.0461 & 0.9320 \\ \end{pmatrix} $ 
\\
\hline
\end{tabular}
\label{gstSummaryTable}
\end{table}

\end{widetext}

\section{Rigorous error correction for single qubits}

Here we apply rigorous measurement error correction to the single-qubit states $|0\rangle$,  $|1\rangle$, and  $|+\rangle$. In each case we compute the Pauli expectation value 
\begin{equation}
\langle Z \rangle = {\rm prob}(0) - {\rm prob}(1)
\end{equation} 
before and after error correction. The results are shown in Table \ref{correctedZTable}. In one case (the $|1\rangle$ state on qubit $Q_1$) negative probabilities were encountered after inverting $\Gamma$, leading to an unphysical result $\langle Z \rangle = -1.0056$; in this case we applied the least-squares minimization protocol (\ref{defLeastSquares}).

\begin{table}[htb]
\centering
\caption{Error corrected Pauli expectation value $\langle Z \rangle$ on ibmq\_essex qubits $\lbrace Q_0, Q_1, Q_2, Q_3, Q_4 \rbrace$ for three prepared states. The $T$ matrix column gives $\langle Z \rangle$ after applying standard $T$ matrix error correction. The $\Gamma$ matrix column gives the results after rigorous measurement correction. The asterisk indicates that unphysical negative probabilities were encountered after the matrix inversion in (\ref{defGammaErrorCorrection}), so the alternative
(\ref{defLeastSquares}) was applied.}
\begin{tabular}{|c|c|}
\hline
$Q_0$ & 
\begin{tabular}{|c|c|c|c|c|}
\hline
State & Ideal & Raw data & $T$ matrix &  $\Gamma$ matrix \\
\hline 
$|0\rangle$ &  1 & 0.9596 & 1.0000 & 0.9576  \\
\hline
$|1\rangle$ &  -1 &-0.8788 & -1.0000 & -0.9698  \\
\hline
$|+\rangle$ &  0 & 0.1301   & 0.0976  & 0.0879  \\
\hline
\end{tabular}
\\ 
\hline
$Q_1$ & 
\begin{tabular}{|c|c|c|c|c|}
\hline
State & Ideal & Raw data & $T$ matrix &  $\Gamma$ matrix \\
\hline 
$|0\rangle$ &  1 & 0.9586 & 1.0000 & 0.9554  \\
\hline
$|1\rangle$ &  -1 & -0.8616 & -1.0000 & -1.0000$^*$  \\
\hline
$|+\rangle$ &  0 & 0.1936  & 0.1595 & 0.1313 \\
\hline
\end{tabular}
\\ 
\hline
$Q_2$ & 
\begin{tabular}{|c|c|c|c|c|}
\hline
State & Ideal & Raw data & $T$ matrix &  $\Gamma$ matrix \\
\hline 
$|0\rangle$ &  1 & 0.9857 & 1.0000 & 0.9847  \\
\hline
$|1\rangle$ &  -1 &-0.8399 & -1.0000 & -0.9677 \\
\hline
$|+\rangle$ &  0 & 0.0728   & -0.0001  & 0.0084  \\
\hline
\end{tabular}
\\ 
\hline
$Q_3$ & 
\begin{tabular}{|c|c|c|c|c|}
\hline
State & Ideal & Raw data & $T$ matrix &  $\Gamma$ matrix \\
\hline 
$|0\rangle$ &  1 & 0.9674 & 1.0000 & 0.9660  \\
\hline
$|1\rangle$ &  -1 & -0.8805 & -1.0000 & -0.9616  \\
\hline
$|+\rangle$ &  0 & 0.1513  & 0.1167 & 0.1147  \\
\hline
\end{tabular}
\\ 
\hline
$Q_4$ & 
\begin{tabular}{|c|c|c|c|c|}
\hline
State & Ideal & Raw data & $T$ matrix &  $\Gamma$ matrix \\
\hline 
$|0\rangle$ &  1 & 0.7017 & 1.0000 & 0.7674  \\
\hline
$|1\rangle$ &  -1 & -0.6802 & -1.0000 & -0.7924  \\
\hline
$|+\rangle$ &  0 & 0.0356 & 0.0360 & 0.0156  \\
\hline
\end{tabular}
\\ 
\hline
\end{tabular}
\label{correctedZTable}
\end{table}

Table \ref{correctedZTable} also compares $\Gamma$ matrix SPAM correction with $T$ matrix SPAM correction. The $T$ matrix, by construction,  always corrects classical states $|0\rangle$ and $|1\rangle$ to their ideal values. Unfortunately this is a defect of the $T$ matrix method and is not physical: To demonstrate this for the $|0\rangle$ state, we use its density matrix as estimated by gate set tomography (the $\rho_{|0\rangle}$ column in Table \ref{gstSummaryTable}) to calculate $\langle Z \rangle = {\rm Tr}(\rho_{|0\rangle}Z)$. As shown in Table \ref{zeroStateTable}, the rigorous results are much closer to the tomographic estimates. 

On the nonclassical $|+\rangle$ state, the $T$ and $\Gamma$ matrix corrections are similar, but differ by 1-3\%. This implies that in current noisy devices, the $T$ matrix, while extremely convenient, has limited accuracy.

\begin{table}[htb]
\centering
\caption{Pauli expectation values of the $|0\rangle$ state compared with gate set tomography (GST). The $T$ matrix and $\Gamma$ matrix columns are copied from Table \ref{correctedZTable}.}
\begin{tabular}{|c|c|c|c|}
\hline
ibmq\_essex qubit   & $\langle Z \rangle \  \ T \ {\rm matrix}$ & $\langle Z \rangle  \ \ \Gamma \ {\rm matrix}$  &  $\langle Z \rangle  \ {\rm GST}$  \\ 
\hline
$Q_0$ & 1.0000  &  0.9576 & 0.9584 \\ 
\hline
$Q_1$ & 1.0000  & 0.9554  & 0.9406 \\ 
\hline
$Q_2$ & 1.0000  &  0.9847 & 0.9857 \\ 
\hline
$Q_3$ & 1.0000  & 0.9660  & 0.9677 \\ 
\hline
$Q_4$ & 1.0000  & 0.7674  & 0.7179 \\ 
\hline
\end{tabular}
\label{zeroStateTable}
\end{table}

We have performed several simulations to assess the accuracy of the single-qubit results: For the first simulation we assume a biased measurement model with 
\begin{equation}
E_0 = 
\begin{pmatrix}
1 & 0 \\
0 & 0.05 \\
\end{pmatrix}
\ \ {\rm and} \ \ 
E_1 = 
\begin{pmatrix}
0 & 0 \\
0 & 0.95 \\
\end{pmatrix},
\label{nonidealPOVMsimulation1}
\end{equation}
representative of the data in Table \ref{gstSummaryTable}, and estimate the effect of sampling errors on the accuracy of the corrected $\langle Z \rangle$ values in Table \ref{correctedZTable}. For a given randomly chosen pure or mixed state $\rho$, we: 
\begin{enumerate}

\item Compute the exact Pauli expectation value
\begin{equation}
\langle Z \rangle_{\rm exact}  = {\rm Tr} (Z \rho).
\end{equation}

\item Simulate the experimental estimation of the probability distribution 
\begin{equation}
p_{\rm noisy}(x) = {\rm Tr} (E_x \rho)
\end{equation}
using $N$ measurement samples and the biased measurement model  (\ref{nonidealPOVMsimulation1}). We write the result of the simulated estimation as ${\hat p}_{\rm noisy}(x)$, which differs from $p_{\rm noisy}(x)$ due to the presence of statistical fluctuations.

\item Apply $\Gamma$ matrix error correction to ${\hat p}_{\rm noisy}(x)$ using the exact $\Gamma$ matrix and constrained least-squares minimization (\ref{defLeastSquares}).

\item Calculate the corrected Pauli expectation value
\begin{equation}
\langle Z \rangle_{\rm corr}  = p_{\rm corr}(0) - p_{\rm corr}(1) .
\end{equation}

\item Calculate the magnitude of the error in the corrected Pauli expectation value, defined as

\begin{equation}
\Delta Z  = \big|  \langle Z \rangle_{\rm corr} - \langle Z \rangle_{\rm exact} \big| . \label{defDeltaZ}
\end{equation}

\end{enumerate}
The results from this protocol, averaged over 1000 density matrices sampled uniformly from the Bloch ball, is given in Table \ref{simulation1Table} for different values of $N$.
 
\begin{table}[htb]
\centering
\caption{Accuracy of $\langle Z \rangle$ error correction versus number of measurement samples $N$. The errors are reported as ${\overline{ \Delta Z}} \pm ( \sigma / {\sqrt N})$, where ${\overline{ \Delta Z}}$ and $\sigma$ are the calculated mean and standard deviation of $\Delta Z$ averaged over 1000 random density matrices. The IBM Q data presented in this paper was taken with 32k measurement samples.}
\begin{tabular}{|c|c|}
\hline
$N$  & $\Delta Z$  \\ 
\hline
1k & $2.4 \! \times \! 10^{-2} \, \pm  \,  6 \! \times \! 10^{-4}$   \\ 
8k & $8.2 \! \times \! 10^{-3}  \,  \pm  \,  2 \! \times \! 10^{-4}$   \\ 
32k & $4.1 \! \times \! 10^{-3}  \, \pm  \,  1 \! \times \! 10^{-4}$  \\ 
100k & $2.4 \! \times \! 10^{-3} \,  \pm   \,  6 \! \times \! 10^{-5}$  \\
\hline
\end{tabular}
\label{simulation1Table}
\end{table}

For the second simulation we repeated the first but with a matrix ${\hat \Gamma}$ that differs from $\Gamma$ due to the presence of statistical fluctuations. Here we assume that $\Gamma$ is obtained not by gate set tomography, but instead by a simulated measurement of the $T$ matrix with perfect initial state preparation and $N$ measurement samples. We find that the effect of using ${\hat \Gamma}$  instead of $\Gamma$ is too small to  affect the results in Table \ref{simulation1Table}.

For the third simulation we fixed $N \! = \! 32000$ (corresponding to the data presented in Tables \ref{gstSummaryTable}-\ref{zeroStateTable}) and allowed for small departures from the biased measurement model. In particular, we assume a family of POVMs of the form
\begin{equation}
E_0 = 
\begin{pmatrix}
1 & \eta \\
\eta & 0.05 \\
\end{pmatrix}
\ \ {\rm and} \ \ 
E_1 = 
\begin{pmatrix}
0 & -\eta \\
-\eta & 0.95 \\
\end{pmatrix},
\label{nonidealPOVMsimulation3}
\end{equation}
where $|\eta| \ll 1$. The results are given in Table \ref{simulation3Table}.

\begin{table}[htb]
\centering
\caption{Accuracy of $\langle Z \rangle$ error correction versus frame tilt $\eta$. The error is reported here as in Table \ref{simulation1Table}, again averaged over 1000 random density matrices.}
\begin{tabular}{|c|c|}
\hline
$\eta$  & $\Delta Z$  \\ 
\hline
$\pm 0.001$  & $4.1 \! \times \! 10^{-3}  \, \pm  \,  1 \! \times \! 10^{-4}$   \\ 
$\pm 0.002$  & $4.5 \! \times \! 10^{-3}  \, \pm  \,  1 \! \times \! 10^{-4}$   \\ 
$\pm 0.005$ & $5.9 \! \times \! 10^{-3}  \, \pm  \,  1 \! \times \! 10^{-4}$   \\ 
\hline
\end{tabular}
\label{simulation3Table}
\end{table}
 
The simulation results imply that the accuracy of $\langle Z \rangle$ error correction, as characterized by (\ref{defDeltaZ}), is about
\begin{equation}
\Delta Z \approx 0.004
\end{equation}
for much of the data presented here. This is significantly smaller than the 1-3\% differences typically observed between $Z$ expectation values corrected by the $T$ and $\Gamma$ matrices.

\section{Conclusion}

In this work we have argued that measurement errors can be fully corrected, with rigorous justification, if two conditions are met:
\begin{enumerate}

\item The noisy $2^n$-outcome POVM elements are diagonal in the classical basis.

\item The $\Gamma$ matrix, defined in (\ref{GammaDefinition}), can be estimated and inverted.

\end{enumerate}
If these conditions are met, the $\Gamma$ matrix should replace the commonly used $T$ matrix \cite{BialczakNatPhys10,NeeleyNat10,DewesPRL12,14114994,160304512,SongPRL17,181102292,190505720,180411326,TannuIEEE19,190411935,190503150,190708518,191000129,191001969,191113289,200104449,200109980}. The correction protocol is given by (\ref{defGammaErrorCorrection}). The data presented in Table \ref{correctedZTable} show that the $T$ and $\Gamma$ matrices produce significantly different results on current superconducting qubits. It will be interesting to apply this technique to other quantum computing architectures as well.

We have also provided tomographic evidence that the noisy measurement operators in {\it single} transmon qubits are nearly diagonal, in support of the biased measurement model for transmon qubits with dispersive readout. This is expected because the dominant measurement error comes from $T_1$ decay during qubit readout.

One approach for extending this technique to multiqubit circuits is to use a tensor product 
\begin{equation}
\Gamma_1 \otimes \Gamma_2 \otimes \cdots \otimes \Gamma_n
\end{equation}
of the single-qubit $\Gamma$ matrices calculated above. However this neglects important multiqubit measurement error correlations \cite{200109980}. A better approach would be to extend the POVM estimation to two or more qubits using gate set tomography or related techniques. This would also determine whether the biased measurement model is valid beyond single transmon qubits, which would require that crosstalk errors are primarily diagonal.

After completing this work we learned that a result equivalent to the biased measurement theorem was obtained by Maciejewski, Zimbor\'as, and Oszmaniec \cite{190708518}, but without addressing state-preparation errors. 

\begin{acknowledgments}

Data was taken using the {BQP} software package developed by the author. We thank Robin Blume-Kohout,  Sasha Korotkov, Filip Maciejewski, and Erik Nielsen for their private communication. We're also grateful to IBM Research for making their devices available to the quantum computing community. This work does not reflect the views or opinions of IBM or any of their employees. 

\end{acknowledgments}

\newpage

\bibliography{/Users/mgeller/Dropbox/bibliographies/algorithms,/Users/mgeller/Dropbox/bibliographies/applications,/Users/mgeller/Dropbox/bibliographies/books,/Users/mgeller/Dropbox/bibliographies/cm,/Users/mgeller/Dropbox/bibliographies/dwave,/Users/mgeller/Dropbox/bibliographies/control,/Users/mgeller/Dropbox/bibliographies/general,/Users/mgeller/Dropbox/bibliographies/group,/Users/mgeller/Dropbox/bibliographies/ions,/Users/mgeller/Dropbox/bibliographies/math,/Users/mgeller/Dropbox/bibliographies/ml,/Users/mgeller/Dropbox/bibliographies/nmr,/Users/mgeller/Dropbox/bibliographies/optics,/Users/mgeller/Dropbox/bibliographies/qec,/Users/mgeller/Dropbox/bibliographies/simulation,/Users/mgeller/Dropbox/bibliographies/software,/Users/mgeller/Dropbox/bibliographies/superconductors,/Users/mgeller/Dropbox/bibliographies/surfacecode,/Users/mgeller/Dropbox/bibliographies/tn,endnotes}

\end{document}